\documentclass{article}
\usepackage[utf8]{inputenc} 
\usepackage{amssymb, amsmath, amsthm, graphicx, subfigure}
\usepackage[dvipsnames]{xcolor}
\usepackage[colorlinks=true,linktoc=page]{hyperref}
\usepackage{braket}
\usepackage{esvect}
\usepackage[margin=1in]{geometry}
\usepackage{mathtools}
\usepackage{tikz}
\usepackage{algorithm}
\usepackage{algorithmicx}
\usepackage[noend]{algpseudocode}

\mathtoolsset{showonlyrefs=true}

\theoremstyle{plain}
\newtheorem{theorem}{Theorem}[section]
\newtheorem{lemma}[theorem]{Lemma}
\newtheorem{corollary}[theorem]{Corollary}

\newtheorem{claim}[theorem]{Claim}
\newtheorem{observation}[theorem]{Observation}

\theoremstyle{definition}
\newtheorem{definition}[theorem]{Definition}

\theoremstyle{remark}
\newtheorem{remark}[theorem]{Remark}
\newtheorem{example}[theorem]{Example}

\newcommand{\CF}{\mathrm{CF}}
\newcommand{\lcf}{\leq_\CF}
\newcommand{\cost}{\mathtt{cost}}

\newcommand{\CSP}{\textsc{CSP}}
\newcommand{\ListCSP}{\textsc{ListCSP}}
\newcommand{\MinCostCSP}{\textsc{MinCostCSP}}

\newcommand{\MaxCSP}{\textsc{MaxCSP}}
\newcommand{\MinCSP}{\textsc{MinCSP}}

\newcommand{\MinUnCut}{\textsc{Min~UnCut}}

\newcommand{\LP}{\mathrm{LP}}
\newcommand{\Opt}{\mathrm{Opt}}
\newcommand{\Pol}{\mathsf{Pol}}

\newcommand{\eq}{\mathrm{eq}}

\hypersetup{colorlinks,
    linkcolor=blue,
    citecolor=ForestGreen,  
    urlcolor=Mahogany,
}

\title{On the Constant-Factor Approximability of Minimum Cost Constraint Satisfaction Problems}
\author{Ian DeHaan\thanks{University of Michigan. Email: \texttt{idehaan@umich.edu} } \and Neng Huang\thanks{University of Michigan. Email: \texttt{nengh@umich.edu}} \and Euiwoong Lee\thanks{University of Michigan. Supported in part by NSF grant CCF-2236669. Email: \texttt{euiwoong@umich.edu} }}

\begin{document}

\maketitle

\begin{abstract}
    We study \emph{minimum cost constraint satisfaction problems} (\MinCostCSP) through the algebraic lens. We show that for any constraint language $\Gamma$ which has the \emph{dual discriminator} operation as a polymorphism, there exists a $|D|$-approximation algorithm for $\MinCostCSP(\Gamma)$ where $D$ is the domain. Complementing our algorithmic result, we show that any constraint language $\Gamma$ where $\MinCostCSP(\Gamma)$ admits a constant-factor approximation must have a \emph{near-unanimity} (NU) polymorphism unless P = NP, extending a similar result by Dalmau et al. on MinCSPs. These results imply a dichotomy of constant-factor approximability for constraint languages that contain all permutation relations (a natural generalization for Boolean CSPs that allow variable negation): either $\MinCostCSP(\Gamma)$ has an NU polymorphism and is $|D|$-approximable, or it does not have any NU polymorphism and is NP-hard to approximate within any constant factor. Finally, we present a constraint language which has a majority polymorphism, but is nonetheless NP-hard to approximate within any constant factor assuming the Unique Games Conjecture, showing that the condition of having an NU polymorphism is in general not sufficient unless UGC fails.
\end{abstract}

\section{Introduction}

Constraint satisfaction problems (CSPs) are a central topic of study in theoretical computer science. In an instance of a CSP, we are given a finite set of variables taking values in a finite domain and a finite set of constraints on these variables, and our goal is to find an assignment to the variables so that all constraints are satisfied. CSPs provide a very expressive framework that encompasses many natural combinatorial problems, including satisfiability problems, graph coloring, and solving linear systems. CSPs in their full generality are NP-hard, and therefore it is natural to consider restrictions which lead to interesting tractable subclasses of CSPs. One very influential type of restrictions is to restrict the \emph{constraint language}, that is, to restrict the set of relations that can be used as constraints. In this line of work, the ultimate goal is to obtain a dichotomy, if it exists, which characterizes the boundary between tractable and NP-hard constraint languages. The first such result, obtained by Schaffer~\cite{schaefer1978complexity}, gave a complete classification of tractable Boolean CSPs. In their landmark paper~\cite{feder1998computational}, Feder and Vardi conjectured that a dichotomy exists for CSPs over general domains. This conjecture, known as the CSP Dichotomy Conjecture, led to a series of fruitful work culminating in the proof of the conjecture obtained independently by Bulatov~\cite{bulatov2017dichotomy} and by Zhuk~\cite{zhuk2020proof}. 

Compared to the standard decision variant of CSPs, there are many natural optimization CSP variants whose complexity landscape is less understood. The most well-studied optimization variant is arguably the maximum constraint satisfaction problem (\MaxCSP), where the objective is to find an assignment that maximizes the number of satisfied constraints. Interesting examples in this class include maximum cut and maximum satisfiability problems. For $\MaxCSP$s, Raghavendra showed that the optimal approximation ratio\footnote{For a maximization (resp. minimization) problem, an approximation algorithm achieves an approximation ratio of $\alpha$, if on an input instance with global optimal $\Opt$, it produces a solution whose objective value is at least $\alpha \cdot \Opt$ (resp. at most $\alpha \cdot \Opt$).} can be obtained by solving and rounding a generic \emph{semidefinite programming} relaxation of the problem~\cite{raghavendra2008optimal} (note that a constant approximation ratio for $\MaxCSP$ can be trivially obtained by uniform random assignment), assuming that the Unique Games Conjecture (UGC)~\cite{khot2002power} holds. However, the exact approximation ratio is not explicit in Raghavendra's result, and the ratios for many interesting $\MaxCSP$s are still open (see e.g., \cite{brakensiek2021mysteries, brakensiek2023separating}). For $\MinCSP$s, the objective is to minimize the number of unsatisfied constraints. $\MinCSP$ can be much harder than the corresponding $\MaxCSP$ in terms of the approximation ratio. In particular, it is at least as hard as the decision problem since an approximation algorithm has to satisfy every constraint when the instance is satisfiable. Valued CSPs generalize $\MinCSP$ by replacing 0-1 constraints with valued constraints, so that for any constraint different partial assignments can incur different costs. Thapper and \v{Z}ivn\'{y} obtained a complexity dichotomy for the task of exact minimization for finite-valued CSPs~\cite{thapper2016complexity}, but the approximability question for this problem is still poorly understood. Ene et al. showed that under some mild technical assumption, there is a generic \emph{linear programming} relaxation for finite-valued CSPs that is optimal for constant-factor approximation, unless UGC fails~\cite{ene2015local}. But unlike for $\MaxCSP$s, it is unknown how to round this linear program. Dalmau~et~al. gave some algebraic conditions which indicate where the boundary of constant-factor approximability (or the lack thereof) for valued CSPs may lie~\cite{dalmau2018towards}, but a full characterization of constant-factor approximability is still unresolved.

In this work, we consider an optimization CSP variant called \emph{minimum cost} CSP (\MinCostCSP). In this variant, assigning any value to a variable comes with a cost, and the cost is a function of the variable-value pair. Our goal is to find a satisfying assignment that minimizes the total cost. \MinCostCSP\ can be seen as a mixed variant between decision and optimization problems, in that we are still required to find a satisfying assignment. It can also be thought of as a special case of valued CSP, where we have some unary constraints representing the variable costs and all other constraints incur 0 cost if satisfied or infinite cost otherwise. \MinCostCSP\ is a very natural CSP variant which avoids the full generality of valued CSPs, yet still includes many interesting problems, such as graph and hypergraph vertex cover, min-ones CSP~\cite{khanna2001approximability} and minimum solution CSP~\cite{jonsson2008introduction}. 

We study the approximability, and in particular constant-factor approximability of \MinCostCSP. Like many aforementioned results, our study is based on the universal-algebraic approach (see e.g.,~\cite{krokhin_et_al:DFU.Vol7.15301.233} for a survey on this approach applied to the exact optimization of valued CSPs), where we investigate the algebraic structure of any constraint language via its \emph{polymorphisms}, which can be thought of certain high-dimensional symmetry that exists in the space of satisfying assignments. More specifically, we seek to algorithmically exploit the existence of desirable polymorphisms or show hardness results based on the lack thereof.

\paragraph{Our contribution}

We obtain constant-factor approximability for constraint languages that have the \emph{dual discriminator} operation as a polymorphism. These constraint languages can be thought of as generalization of 2-SAT to arbitrary finite domains. We give two algorithms for this class of problems, one using a greedy approach and the other based on a natural linear programming relaxation of the problem. Both algorithms crucially use the consistency notion called \emph{(2,3)-minimality}~\cite{barto2014collapse}.

\begin{theorem}\label{thm:intro_algo}
    Let $\Gamma$ be a constraint language over some domain $D$ that has the dual discriminator operation as its polymorphism. Then $\MinCostCSP(\Gamma)$ can be $|D|$-approximated in polynomial time.
\end{theorem}

Complementing our algorithmic result, we obtain the following hardness condition which says that constant-factor approximation is NP-hard for any constraint language which does not have a \emph{near-unanimity} (NU) polymorphism. 

\begin{theorem}\label{thm:intro_necessary}
    Let $\Gamma$ be a constraint language such that $\MinCostCSP(\Gamma)$ has a constant-factor approximation, then $\Gamma$ has a NU polymorphism, unless P = NP.
\end{theorem}

Near unanimity operations are well-studied in universal algebra (see e.g., ~\cite{baker1975polynomial}), and they have also appeared in the study of CSPs~\cite{feder1998computational,dalmau2018towards,dalmau2019robust}. In particular, Dalmau et al. showed that for valued CSPs the existence of NU polymorphisms is also a necessary condition for constant-factor approximability~\cite{dalmau2018towards}. It can be verified that for \MinCostCSP\ over the Boolean domain, the condition of having an NU polymorphism is not only necessary but also sufficient for constant-factor approximability (see Remark~\ref{remark:boolean_nu} for more discussion). However, as soon as the domain has at least 3 elements, there exist constraint languages which have NU polymorphisms yet does not admit constant-factor approximation, unless the UGC fails. We present such an example in Section~\ref{sec:counter_example}. 

Finally, as an application of our hardness and algorithmic results, we fully classify the constant-factor approximability of constraint languages that include all \emph{permutation relations}, showing that the existence of an NU polymorphism is also a sufficient condition for this class. These languages can be thought of as a natural generalization of Boolean CSPs where we are allowed to apply constraints to negated variables. Our classification relies on the classification of \emph{homogeneous algebras} by Marchenkov~\cite{marchenkov1982homogeneous}.

\begin{theorem}\label{thm:intro_perm}
    Let $\Gamma$ be a constraint language over some domain $D$ that contains all permutation relations over $D$. Then $\MinCostCSP(\Gamma)$ can be $|D|$-approximated if $\Gamma$ has an NU polymorphism, and is NP-hard to approximate within any constant factor otherwise.
\end{theorem}

\paragraph{Related work}
The approximability of \MaxCSP, \MinCSP, as well as $\MinCostCSP$ over the Boolean domain was fully classified by Khanna et al.~\cite{khanna2001approximability}. In particular, for $\MinCostCSP$ they obtained the following complete classification: $\MinCostCSP(\Gamma)$ can be solved to optimality in polynomial time if $\Gamma$  is ``width-2 affine'', that is, $\Gamma$ can be expressed as a conjunction of linear equations over $\mathbb{F}_2$ where each equation has at most 2 variables; the problem can be approximated within a constant factor in polynomial time if $\Gamma$ can be expressed as a 2CNF-formula, or if $\Gamma$ is IHB-B$+$ (expressible as a CNF formula where each clause is of the form $x_1 \vee \cdots \vee x_k$, $\neg x_1 \vee x_2$, or $\neg x_1$ where $k \leq K$ for some $K$ depending on $\Gamma$), or if $\Gamma$ is IHB-B$-$ (defined analogously to IHB-B$+$ with every literal replaced by its negation). Otherwise, $\MinCostCSP(\Gamma)$ is NP-hard to approximate within any constant factor. (See Remark~\ref{remark:boolean_nu} for a more detailed discussion in the context of our results.)

Over the general domain, a dichotomy for solving $\MinCostCSP(\Gamma)$ optimally was obtained by Takhanov~\cite{takhanov2010dichotomy}. Takhanov's characterization is based on local algebraic conditions satisfied by polymorphisms of $\Gamma$.

Kumar et al. showed that for a large class of covering and packing problems that can be expressed as $\MinCostCSP$ (they called it ``Strict-CSP'') over the general domain, a generic linear programming relaxation gives the optimal approximation ratio achievable in polynomial time, assuming the Unique Games Conjecture~\cite{kumar2011lp}. Their result generalizes the earlier UGC-based hardness results for vertex cover~\cite{khot2008vertex} and the $k$-uniform hypergraph vertex cover problems~\cite{bansal2010inapproximability}. 

An important special case for $\MinCostCSP(\Gamma)$ where $\Gamma$ consists of one single binary relation has been studied in the literature under the name ``min-cost graph homomorphism'' (in the case where the binary relation is symmetric) or ``min-cost di-graph homomorphism'' (in the general binary case). Dichotomy results for optimally solving these problems are known based on graph-theoretic properties~\cite{gutin2008dichotomy, hell2012dichotomy}.  Hell et al. gave a similar dichotomy for constant-factor approximability for the min-cost graph homomorphism problem in the case where the graph (equivalently, the binary relation) is reflexive (every vertex has a self-loop) or irreflexive (no vertex has a self-loop)~\cite{hell2012approximation}.\footnote{We note that an ICALP'19 paper~\cite{rafiey2019toward} claimed that the following dichotomy for constant-factor approximability holds over all (undirected) graphs: either a graph $G$ has a conservative majority polymorphism and is constant-factor approximable, or it does not and is NP-hard to approximate within any constant factor. Our Theorem~\ref{thm:counterexample} contradicts this claim under the Unique Games Conjecture and $P \neq NP$. }

Another \CSP\ variant closely related to \MinCostCSP\, is  \ListCSP\ which can be thought of as a special case of \MinCostCSP\ where the costs take values in $\{0, \infty\}$. Bulatov obtained a complete classification for this problem (under the name ``conservative CSP'') based on the algebraic approach~\cite{bulatov2003tractable} (see also~\cite{barto2011dichotomy, bulatov2016conservative}).

\paragraph{Organization of the paper}
The rest of the paper is organized as follows. In Section~\ref{sec:prelim}, we formally define the problems and introduce some algebraic concepts that are needed throughout the paper. In Section~\ref{sec:algo_dd}, we present our main algorithmic results, proving Theorem~\ref{thm:intro_algo}. In Section~\ref{sec:nu_necessary}, we prove some algebraic conditions sufficient for reductions between $\MinCostCSP$s, and use them to prove Theorem~\ref{thm:intro_necessary}. Finally, in Section~\ref{sec:application}, we use our results to give a dichotomy of constant-factor approximability for $\MinCostCSP$s that contain all permutation relations, proving Theorem~\ref{thm:intro_perm}.

\section{Preliminaries}\label{sec:prelim}
\subsection{CSP, ListCSP, and MinCostCSP}
Let $D$ be a finite set. A \emph{relation} over $D$ is a subset $R \subseteq D^k$ for some positive integer $k$, where $D$ is called the \emph{domain} of $R$ and $k$ is called the \emph{arity} of $R$. A set of relations $\Gamma$ over the same domain $D$ is called a \emph{constraint language}. Throughout this paper, any constraint language we consider will be assumed to contain finitely many relations whose common domain will be denoted by $D$. The elements in $D$ will be referred to as \emph{labels}.
\begin{definition}
    Let $\Gamma$ be a constraint language. An instance of $\CSP(\Gamma)$ is a tuple $I = (V, \mathcal{C})$, where $V$ is a finite set of variables and $\mathcal{C}$ a finite set of constraints. Each constraint $C \in \mathcal{C}$ is of the form $(R, S)$, where $R$ is a relation in $\Gamma$ with some arity $k$ and $S \in V^k$ a $k$-tuple of variables.
    An \emph{assignment} for $I$ is a function $A: V \to D$. We say that $A$ satisfies a constraint $C = (R, (x_1, \ldots, x_k))$ if $(A(x_1), \ldots, A(x_k)) \in R$, and we say that $A$ is a satisfying assignment for $I = (V, \mathcal{C})$ if $A$ satisfies every constraint in $\mathcal{C}$. 
\end{definition}

One closely related variant of $\CSP$ is the following $\ListCSP$ problem. This problem has also been referred to as \emph{conservative} $\CSP$ in the literature.

\begin{definition}
    Let $\Gamma$ be a constraint language. An instance of $\ListCSP(\Gamma)$ is a tuple $I = (V, \mathcal{C}, \{R_x\}_{x \in V})$, where $V$ is a finite set of variables and $\mathcal{C}$ a finite set of constraints, as in the definition of $\CSP(\Gamma)$. In addition to $V$ and $\mathcal{C}$, we are also given a subset $R_x \subseteq D$ for every variable $x \in V$. We say that $A$ is a satisfying assignment for $I = (V, \mathcal{C}, \{R_x\}_{x \in V})$ if $A$ satisfies every constraint in $\mathcal{C}$ and $A(x) \in R_x$ for every $x \in V$. 
\end{definition}

$\ListCSP$ may be equivalently viewed as the ordinary \CSP\ where all unary constraints are allowed in addition. In this paper, we will mainly focus on $\MinCostCSP$, which further generalizes $\ListCSP$.

\begin{definition}
    Let $\Gamma$ be a constraint language. An instance of $\MinCostCSP(\Gamma)$ is a tuple $I = (V, \mathcal{C}, \cost)$, where $V$ and $\mathcal{C}$ are the same as they are in the definition of $\CSP(\Gamma)$, and we are also given a function $\cost: V \times D \to \mathbb{R}^{\geq 0} \cup \{+\infty\}$. For any assignment $A$ for $I$, the cost of $A$ is defined to be $\cost_I(A) = \sum_{x \in V}\cost(x, A(x))$. The goal is find a satisfying assignment with the minimum cost.
\end{definition}

Given a $\MinCostCSP$ instance $I$, the optimum cost is denoted by $\Opt(I) := \min_A \cost_I(A)$, where $A$ ranges over all satisfying assignments for $I$. Although we allowed infinite costs in the definition of $\MinCostCSP$, this is not essential in the context of approximability. In fact, we may simulate infinite cost by setting costs to be prohibitively high, say larger than $|V|$ times the maximum of any other finite cost, so that no approximation algorithm will pick the label. The inclusion of the infinite cost conveniently allows us to assume without loss of generality that $\Gamma$ contains all unary relations.

\begin{observation}\label{obs:mincost_unary}
    Let $\Gamma$ be a constraint language over some domain $D$ and $\Gamma' = \Gamma \cup \{S \subseteq D \mid S \neq \varnothing\}$. Then $\MinCostCSP(\Gamma)$ and $\MinCostCSP(\Gamma')$ are equivalent. Namely, any instance of $\MinCostCSP(\Gamma)$ can be solved as an instance of $\MinCostCSP(\Gamma')$ and vice versa.
\end{observation}
\begin{proof}
    One direction is clear since $\Gamma \subseteq \Gamma'$. For the other direction, for any unary constraint $S(x)$ where $S \subseteq D$ (the constraint which says the label of $x$ must be in $S$), we simply set $\cost(x, a) = \infty$ for any $a \not\in S$. 
\end{proof}

Observation~\ref{obs:mincost_unary} implies that $\MinCostCSP(\Gamma)$ contains $\ListCSP(\Gamma)$ as a special case. In particular, if $\ListCSP(\Gamma)$ is NP-hard, then for $\MinCostCSP(\Gamma)$ finding any satisfying assignment regardless of cost is NP-hard as well.

\subsection{Polymorphisms}

\begin{definition}
    Let $f: D^k \to D$ be a $k$-ary operation on the domain $D$, and $R$ some $m$-ary relation over the same domain. We say that $f$ preserves $R$ (or $f$ is a \emph{polymorphism} of $R$) if for every $(a_{1,1}, \ldots, a_{1,m}), \ldots, (a_{k,1}, \ldots, a_{k,m}) \in R$ we have $(b_1, \ldots, b_m) \in R$, where $b_i = f(a_{1, i}, \ldots, a_{k, i})$ for every $i \in [m]$. Given a constraint language $\Gamma$ we say that $f$ preserves $\Gamma$ (or $f$ is a \emph{polymorphism} of $\Gamma$) if $f$ preserves every $R \in \Gamma$. We use $\Pol(R)$ to denote the set of all polymorphisms of $R$, and (abusing notation slightly) $\Pol(\Gamma) = \bigcap_{R \in \Gamma} \Pol(R)$ to denote the set of all polymorphisms of $\Gamma$.
\end{definition}

It follows directly from the definition that for any $n \geq 1$, $i \in [n]$, the projection operation $\mathrm{proj}^n_i: D^n \to D, (x_1, \ldots, x_n) \mapsto x_i$ is a polymorphism. Also, any composition of polymorphisms is still a polymorphism. A set of operations satisfying these two properties is known as a \emph{clone}.

\begin{definition}\label{def:clone}
    A \emph{clone} over $D$ is a set of operations $\mathcal{F} \subseteq \bigcup_{n \geq 1}\{f: D^n \to D\}$ for which the following holds:
    \begin{itemize}
        \item $\mathcal{F}$ contains the projection operation $\mathrm{proj}^n_i$ for every $n \geq 1$, $i \in [n]$.
        \item For every $m$-ary $g \in \mathcal{F}$ and $n$-ary $f_1, \ldots, f_m \in \mathcal{F}$, the $n$-ary function $g'$ defined by
        \[
        g': (x_1, \ldots, x_n) \mapsto g(f_1(x_1, \ldots, x_n), \ldots, f_m(x_1, \ldots, x_n))
        \]
        is also in $\mathcal{F}$.
    \end{itemize}
\end{definition}

Clones have been studied extensively in the universal algebra community. Many complexity-theoretic classification results for CSPs depend on classification results for their corresponding family of polymorphism clones (see e.g.~\cite{barto_et_al:DFU.Vol7.15301.1} for more on polymorphism clones). 

We say that a function $f: D^k \to D$ is \emph{conservative}, if $f(x_1, \ldots, x_k) \in \{x_1, \ldots, x_k\}$ for every $x_1, \ldots, x_k \in D$. It is immediate from the definition that if $\Gamma$ is constraint language that contains all unary relations, then every $f \in \Pol(\Gamma)$ is conservative. 

\begin{definition}
    Let $f: D^k \to D$ be a $k$-ary operation for some $k \geq 3$. We say that $f$ is a \emph{near-unanimity (NU)} operation if for every $a, b \in D$, 
    \[
    f(a, a, \ldots, a, b) = f(a, a, \ldots, b, a) = \cdots = f(b, a, \ldots, a, a) = a.
    \]
    In other words, if all but one inputs are equal to some $a \in D$, then $f$ outputs $a$. If $k = 3$, then $f$ is also called a \emph{majority} operation.
\end{definition}

\begin{example}
    Let \emph{the dual discriminator operation} $d: D^3 \to D$ be defined by
    \[
    d(x_1, x_2, x_3) = \left\{\begin{array}{ll}
        a & \text{if } |\{i \mid x_i = a\}| \geq 2, \\
        x_1 & \text{otherwise.}
    \end{array}\right.
    \]
    Then $d$ is a majority operation. This operation will be featured heavily in our algorithmic results.
\end{example}

Constraint languages that are preserved by some NU operation enjoy the following nice property.

\begin{definition}
    Let $R$ be a $n$-ary relation. Let $S_{n, k} = \{(i_1, \ldots, i_k) \mid 1 \leq i_1 < \cdots < i_k \leq n\}$ be the set of $k$-tuples whose entries are in $[n]$ and in increasing order. For any $s = (i_1, \ldots, i_k) \in S_{n,k}$ and $x = (x_1, \ldots, x_n)$, let $\Pi_{s} x = (x_{i_1}, \ldots, x_{i_k})$ and $\Pi_{s}R = \{\Pi_{s} x \mid x \in R\}$. We say that $R$ is $k$-decomposable if for every $x = (x_1, \ldots, x_n)$, we have
    \[
    x \in R \quad\Leftrightarrow\quad \forall s \in S_{n, k}, \Pi_s x \in \Pi_s R.
    \] 
    We say that $\Gamma$ is $k$-decomposable if every $R \in \Gamma$ is $k$-decomposable.
\end{definition}

\begin{theorem}[Theorem 3.5 in~\cite{jeavons1998constraints}]\label{thm:nu_decompose}
    Let $\Gamma$ be a constraint language and $k \geq 2$. If $\Gamma$ is preserved by some $(k+1)$-ary NU operation, then $\Gamma$ is $k$-decomposable.
\end{theorem}

\section{Algorithms for the dual discriminator}\label{sec:algo_dd}
In this section, we give two $|D|$-approximation algorithms for $\MinCostCSP(\Gamma)$ where $\Gamma$ is a constraint language over $D$ preserved by the dual discriminator operation, thus proving Theorem~\ref{thm:intro_algo}. Both algorithms will assume that the input instance is a satisfiable \emph{binary}, \emph{(2,3)-minimal} instance. We introduce and justify this assumption in Section~\ref{subsec:binary_23minimal}, and then present a greedy algorithm in Section~\ref{subsec:greedy} and a LP-based algorithm in Section~\ref{subsec:LP}.

\subsection{Reducing to satisfiable binary (2,3)-minimal instances}\label{subsec:binary_23minimal}
A CSP instance $I = (V, \mathcal{C})$ is \emph{binary}, if every constraint in $\mathcal{C}$ has arity at most 2. Theorem~\ref{thm:nu_decompose} allows us to assume that the input instance is binary, since it is preserved by the dual discriminator operation which is a 3-ary NU operation. For any binary instance $I$, we may also write it as a triple $I = (V, \{R_u\}_{u \in V}, \{R_{u, v}\}_{u\neq v\in V})$, where we have one unary relation $R_x$ for every $u \in V$ and one binary relation $R_{u, v}$ for every pair $(u,v) \in V$. Note that to write $I$ in this form we may take the intersection of relations with the same scope or add complete relation on some variable(s) if none exists.

\begin{definition}[See e.g., \cite{barto_et_al:DFU.Vol7.15301.1}]\label{def:2-3-minimal}
    A binary CSP instance $I = (V, \{R_u\}_{u \in V}, \{R_{u, v}\}_{u\neq v \in V})$ over domain $D$ is \emph{$(2,3)$-minimal} if 
    \begin{enumerate}
        \item[(a)] For every distinct $u, v\in V$, $R_{v,u} = \{(b, a) \mid (a, b) \in R_{u, v}\}$. 
        \item[(b)] For every distinct $u, v \in V$, $R_u = \{a \in D \mid \exists b \in D, (a, b) \in R_{u,v}\}$.
        \item[(c)] For every pairwise distinct $u,v,w \in V$ and $(a, b) \in R_{u,v}$, there exists $c \in R_w$ such that $(a, c) \in R_{u,w}$ and $(b, c) \in R_{v,w}$.
    \end{enumerate}
\end{definition}

Informally, the definition says that $I$ is (2,3)-minimal if any partial satisfying assignment (that is, a partial assignment that does not immediately falsify any constraint) to two variables can be extended to a partial satisfying assignment to three variables.
Given a binary CSP instance $I = (V, \{R_u\}_{u \in V}, \{R_{u, v}\}_{u\neq v \in V})$, we may transform it into a $(2,3)$-minimal instance using the following procedure:
\begin{center}
\begin{algorithmic}
    \Repeat
        \For{every distinct $u, v, w \in V$}
            \State $R_{u, v} \gets \{(a, b) \in R_{u,v} \mid \exists c \in D, (a, c) \in R_{u, w} \wedge (b, c) \in R_{v, w}\} $
            \State $R_{u, v} \gets \{(a, b) \in R_{u,v} \mid (b, a) \in R_{v, u}\} $
            \State $R_u \gets \{a \in D \mid \exists b \in D, (a, b) \in R_{u, v}\}$
            \State $R_v \gets \{b \in D \mid \exists a \in D, (a, b) \in R_{u, v}\}$
        \EndFor
    \Until{none of the relations change}
\end{algorithmic}
\end{center}

Clearly, when the above procedure stops, the instance must be $(2,3)$-minimal (otherwise the loop would have continued). The new instance may have constraints that are not in our original constraint language. However, since the new constraints are all obtained by pp-definitions from the original constraint language, they are preserved by the same polymorphisms (see Definition~\ref{def:pp_def} and Theorem~\ref{thm:galois}). In particular, the new instance is still preserved by the dual discriminator operation.

We say that a $(2,3)$-minimal instance is \emph{trivial} if at least one of the unary relations is empty, and it is \emph{nontrivial} otherwise. A trivial $(2,3)$-minimal instance has no satisfying assignment since no assignment can satisfy an empty relation. On the other hand, if the instance is nontrivial and its constraint language has \emph{bounded width}, then such an instance is always guaranteed to have a satisfying assignment which we can find in polynomial time.
\begin{theorem}[\cite{barto2014collapse,barto2014constraint}]\label{thm:bounded-width}
    Let $I$ be a nontrivial $(2,3)$-minimal instance whose constraint language has bounded width. Then $I$ has a satisfying assignment which can be found in polynomial time.
\end{theorem}

We will not define the term ``bounded width'' formally here (see Theorem~\ref{thm:unbounded-abelian} for a characterization). For us it is sufficient to note that any constraint language that is preserved by an NU polymorphism has bounded width~\cite{feder1998computational}, so any unsatisfiable instance will necessarily give rise to a trivial (2,3)-minimal instance. From now on we will assume that the (2,3)-minimal instance is nontrivial, and therefore satisfiable. Also, observe any satisfying assignment for the original instance will remain satisfying for the new (2,3)-minimal instance. This means any assignment we obtain on the new instance will give an approximation ratio at least as good on the original instance, and therefore we may shift our attention to this new instance instead.

We crucially use the following characterization for binary relations preserved by the dual discriminator operation. The same characterization was used by~\cite{dalmau2019robust} in the context of robust satisfiability (these relations have also been studied under the name of \emph{0/1/all relations}~\cite{cooper1994characterising}). We sketch a short proof here for completeness.

\begin{lemma}[See e.g.~\cite{barto_et_al:DFU.Vol7.15301.1, dalmau2019robust}]\label{lem:01all}
    Let $R \subset D^2$ be a binary relation. Then $R$ is preserved by the dual discriminator operation $d: D^3 \to D$ if and only if $R$ is of one of the following forms:
    \begin{itemize}
        \item $R = P \times Q$ for some $P, Q \subseteq D$.
        \item $R = (\{u\} \times Q) \cup (P \times \{v\})$ for some $u \in P \subseteq D$ and $v \in Q \subseteq D$.
        \item $R = \{(u, \pi(u)) \mid u \in P\}$ for some $P,Q \subseteq D$ and bijective $\pi: P \to Q$.
    \end{itemize}
\end{lemma}
\begin{proof}
    It is easy to verify that if $R$ is one of these three types then it is preserved by $d$. Let us prove the other direction. Let $P = \{x \in D \mid \exists y \in D \text{ s.t. } (x, y) \in R\}$ and $Q = \{y \in D \mid \exists x \in D \text{ s.t. } (x, y) \in R\}$. For any $u \in P$, if there exist two distinct $v_1, v_2 \in Q$ such that $(u, v_1), (u, v_2) \in R$, then for any $v \in Q$, we have also $(u, v) \in R$. In other words, we have $\{u\} \times Q \subseteq R$. This is because we can find some $(u_1, v)\in R$ by the definition of $Q$, and apply $f$ to the three pairs $(u_1, v), (u, v_1), (u, v_2) \in R$ to obtain $(u, v)$ (observe that $f(u_1, u, u) = u$ and $f(v, v_1, v_2) = v$).

    Now if we have two distinct $u_1, u_2 \in P$ such that $\{u_1\} \times Q \subseteq R$ and $\{u_2\} \times Q \subseteq R$, then for every $v \in Q$ there are two distinct $u_1, u_2$ such that $(u_1, v), (u_2, v)\in R$. By applying the above argument with $P$ and $Q$ reversed  we get that for every $v \in Q$, $P \times \{v\} \subseteq R$, so it must be the case that $R = P \times Q$. 
    
    Now assume that there exists exactly one $u \in P$ such that $\{u\} \times Q \subseteq R$ and $R \neq P \times Q$. Then there exists $(u_1, v) \in R$ such that $u_1 \neq u$, so we have $P \times \{v\} \subseteq R$. Note that in this case we must have $R = (\{u\} \times Q) \cup (P \times \{v\})$, for the existence of any other pair would imply that $R = P \times Q$.

    Finally, if $R \neq P \times Q$ and there is no $u \in P$ such that $\{u\} \times Q \subseteq R$, then there can also be no $v \in Q$ such that  $P \times \{v\} \subseteq R$. So for every $u \in P$ there is a unique $v$ such that $(u, v) \in R$, and we can find some bijective $\pi: P \to Q$ such that $R = \{(u, \pi(u)) \mid u \in P\}$.
\end{proof}

\subsection{A greedy algorithm}\label{subsec:greedy}
We now present a greedy algorithm which is a generalization of a 2-approximation algorithm for MinOnes 2-SAT due to Gusfield and Pitt~\cite{gusfield1992bounded}. In this algorithm, we will greedily pick labels for the variables one by one, and each label we pick may potentially restrict the set of feasible labels for some other variables. For general CSPs, this restriction can be rather arbitrary and difficult to control. However, for binary $(2,3)$-minimal instances preserved by the dual discriminator operation, the restriction is very simple: either the set of feasible labels is unchanged, or it is restricted to a singleton set (as is guaranteed by Lemma~\ref{lem:01all}). This means for variables where proper restriction happens we can simply fix it to the label in the singleton set. On the other hand, the variables whose set of feasible labels didn't change induce a sub-instance of the original instance and the algorithm may recurse on this sub-instance. This property guarantees that we will always produce a satisfying assignment, as is formalized in the following lemma.

\begin{definition}
    Let $I = (V, \{R_u\}_{u \in V}, \{R_{u,v}\}_{u,v\in V})$ be a (2,3)-minimal instance. For any $u, v \in V$, $a \in R_u$ and $b \in R_v$, we say that $v$ is fixed to $b$ by assigning $a$ to $u$ if $u = v$ and $a = b$, or $u \neq v$ and $R_{u, v} \cap (\{a\} \times R_v) = \{(a, b)\}$. In other words, $b$ is the only feasible label left for $v$ if we assign the label $a$ to $u$.
\end{definition}

\begin{lemma}\label{lem:greedy_correctness}
    Let $I = (V, \{R_u\}_{u \in V}, \{R_{u,v}\}_{u,v\in V})$ be a nontrivial (2,3)-minimal instance preserved by the dual discriminator operation. Let $u_0 \in V$, $a \in R_{u_0}$, $S$ be the set of variables fixed by assigning $a$ to $u_0$. For every $v \in S$, let $A_S(v)$ be the unique label that $v$ is fixed to by assigning $a$ to $u_0$. Let $A': V \backslash S \to D$ be any satisfying assignment for the induced (2,3)-minimal instance $I' = (V \backslash S, \{R_u\}_{u \in V \backslash S}, \{R_{u,v}\}_{u,v\in V \backslash S})$, then 
    \[
    A: V \to D, \quad A(v) = \left\{\begin{array}{ll}
        A_S(v) & \text{if } v \in S, \\
        A'(v) & \text{otherwise}
    \end{array}\right.
    \]
    is a satisfying assignment for $I$.
\end{lemma}
\begin{proof}
    Clearly all unary constraints in $I$ are satisfied. Any binary constraint $R_{u, v}$ where $u, v \in V \backslash S$ is satisfied since $A'$ is a satisfying assignment. The remaining two cases, where $u \in S, v \in V \backslash S$ or $u, v \in S$, follow from the following claim:

    \begin{claim}
        For every distinct $v \in S$ and $w \in V$, if $(a, c) \in R_{u_0, w}$, then we have $(A_S(v), c) \in R_{v, w}$.
    \end{claim}
    \begin{proof}
        This claim clearly holds for $v = u_0$. Suppose $v \neq u_0$. Since $(a, c) \in R_{u_0, w}$, by (2,3)-minimality we can find $b \in R_v$ such that $(a, b) \in R_{u_0, v}, (b, c) \in R_{v, w}$. This $b$ must coincide with $A_S(v)$, since $v$ is fixed by assigning $a$ to $u_0$. Thus, we have $(A_S(v), c) \in R_{v, w}$.
    \end{proof}
    For $u \in S, v \in V \backslash S$, since $(a, A'(v)) \in R_{u_0, v}$, we may apply the claim and obtain that $(A_S(u), A'(v)) \in R_{u, v}$, so this constraint is satisfied. For $u, v \in S$, again since since $(a, A_S(v)) \in R_{u_0, v}$, we may apply the claim and obtain that $(A_S(u), A_S(v)) \in R_{u, v}$, so this constraint is also satisfied. It follows that all binary constraints in $I$ are satisfied by $A$, so it is indeed a satisfying assignment for $I$.
\end{proof}

To guarantee a constant-factor approximation, we need to pick labels in a clever way. One naive idea is to compute the total cost incurred by variables fixed by each label and pick the label that minimizes this cost. However, this does not work because the optimum assignment may incur more cost but save by fixing more variables. One way to fix this naive idea is to consider a \emph{derived cost} $t$, instead of the original cost $\cost$. In particular, when we face a decision for some variable $u$ (i.e., $|R_u| \geq 2$), we pick the label that minimizes the total derived cost of fixed variables. We then pay $|R_u| \leq |D|$ times this cost towards reducing the derived cost of \emph{all} variable-label pairs that are or could've been fixed by assignment to $u$. The key idea here is that one of the labels in $R_u$ will be taken by the optimal assignment, so the amount we pay are always within a factor of $|D|$ from the total cost of the optimal assignment.

We now formally present the algorithm and its analysis. The pseudocode for the algorithm can be found in Algorithm~\ref{alg:greedy}.

\begin{algorithm}

\caption{Greedy algorithm for \MinCostCSP}\label{alg:greedy}
 \hspace*{\algorithmicindent} \textbf{Input:} $I = (V, \{R_u\}_{u \in V}, \{R_{u,v}\}_{u,v\in V})$ a (2,3)-minimal instance,  $\cost: V \times D \to \mathbb{R} \cup \{+\infty\}$ \\
 \hspace*{\algorithmicindent} \textbf{Output:}  $A: V \to D$ a satisfying assignment for $I$. \\
 \hspace*{\algorithmicindent}  \textbf{Initialization:} $t: V 
 \times D \to \mathbb{R}^{\geq 0}\cup\{+\infty\}, t(u, a) \gets \cost(u, a)$
\begin{algorithmic}[1]
\Procedure{MinCostHom-Greedy}{$I$}
\State $S \gets \{u \in V \mid |R_u| \geq 2\}$ \Comment{The set of undetermined variables}
\For {$u \in V \backslash S$}
    \State $A(u) \gets$ the unique label $a$ in $R_u$
    \State $t(u, a) \gets 0$ \label{line:zero_t_outside_loop}
\EndFor
\While {$S \neq \varnothing$} \label{line:while_loop}
    \State Pick an arbitrary $v \in S$
    \For {$a \in R_v$}
        \State $F_{a} \gets \{u \in S \mid u \text{ is fixed by assigning } v = a\}$
        \For {$u \in F_a$}
            \State $A_a(u) \gets $ the label $u$ is fixed to if assigning $v = a$
        \EndFor
        \State $c_a \gets \sum_{u \in F_a} t(u, A_a(u))$ \Comment{Calculate the cost for each label $a$ if we assign it to $v$}
    \EndFor
    \State $a_0 \gets \arg\min_{a \in R_v} c_a$
    \State Initialize an array $\Delta t$ indexed by $V \times D$ to all zeros. 
    \State \Comment{$\Delta t$ stores the amount by which we will decrease $t$ this round.}
    \State \Comment{Intuitively, this is the cost we need to pay in this round. }
    \For {$a \in R_v$} 
        \State Find $0 \leq t_a(u) \leq t(u, A_a(u))$ for every $u \in F_a$ such that $\sum_{u \in F_a} t_a(u) = c_{a_0}$ \label{line:find_t}
        \State \Comment{Always possible since $\sum_{u \in F_a} t(u, A_a(u)) = c_a \geq c_{a_0}$}
        \State \Comment{In particular, $t_{a_0}(u) = t(u, A_{a_0}(u))$ for every $u \in F_{a_0}$}
        \For {$u \in F_a$}
            \State $\Delta t(u, A_a(u)) \gets \max(\Delta t(u, A_a(u)), t_a(u))$ \Comment{Choose the largest decrease for $t$} \label{line:find_Delta_t} 
        \EndFor
    \EndFor
    
    \For {$u \in V$} \label{line:decrease_t_for_loop}
        \For {$a \in D$}
            \State $t(u, a) \gets t(u, a) - \Delta t(u, a)$ \Comment{Update the costs} \label{line:zero_t_inside_loop}
        \EndFor
    \EndFor
    \For {$u \in F_{a_0}$}
        \State $A(u) \gets A_{a_0}(u)$ \Comment{Set values to the variables fixed by assigning $a_0$ to $u$}
    \EndFor
    \State $S \gets S \backslash F_{a_0}$ \Comment{Remove variables fixed in this step from the set of undetermined variables}
\EndWhile
\State \Return $A$
  
\EndProcedure
\end{algorithmic}
\end{algorithm}

We start by making the following observations.

\begin{observation}\label{obs:greedy}
    Let $t^{\mathrm{end}}$ be the values of $t$ when $A$ is returned. For every $u \in V$, $t^{\mathrm{end}}(u, A(u)) = 0$.
\end{observation}
\begin{proof}
    For variables fixed outside the while loop, this follows from Line~\ref{line:zero_t_outside_loop}. For variables that obtained its assignment inside the while loop, this follows from Lines~\ref{line:find_t} and~\ref{line:zero_t_inside_loop} (since $\Delta t(u, A_{a_0}(u)) = t(u, A_{a_0}(u))$ for every $u \in F_{a_0}$).
\end{proof}

\begin{observation}\label{obs:monotone}
    The value of $t$ never increases during the algorithm.  In particular, for every $u \in V$ and $a \in D$, we have  $\cost(u, a) - t^{\mathrm{end}}(u, a) \geq 0$.
\end{observation}
\begin{proof}
    This is because $\Delta t$ is always nonnegative.
\end{proof}

\begin{observation}\label{obs:nonnegative}
    For every $u \in V$ and $a \in D$, we have $t(u, a) \geq 0$ at any point of the algorithm.
\end{observation}
\begin{proof}
    This is because when an update happens at Line~\ref{line:find_Delta_t}, the updated value for $\Delta t(u, A_a(u))$ is equal to $t_a(u) \leq t(u, A_a(u))$. 
\end{proof}

The guarantee on the approximation ratio is based on the following claim.

    \begin{claim}\label{claim:approx_ratio}
        Let $A': V \to D$ be any satisfying assignment for $I$. For each iteration of the while loop beginning at Line~\ref{line:while_loop}, let $t^{\mathrm{before}}, t^{\mathrm{after}}$ be the values of $t$ at the beginning and the end of the iteration respectively. Then we have 
        \[
        \sum_{u \in V, a \in D}(t^{\mathrm{before}}(u, a) - t^{\mathrm{after}}(u, a)) \leq |D| \cdot \sum_{u \in V}(t^{\mathrm{before}}(u, A'(u)) - t^{\mathrm{after}}(u, A'(u)))
        \]
    \end{claim}
    \begin{proof}
        Let the variables be defined as they are in the algorithm. For every $u \in V$ and $a \in D$ define the auxiliary variables
        \[
        \Delta'(u, a) = \sum_{a': a' \in R_v, u \in F_{a'}, a = A_{a'}(u)} t_{a'}(u).
        \]
        
        Note that the quantity $(t^{\mathrm{before}}(u, a) - t^{\mathrm{after}}(u, a))$ is nonzero only if there exists some $a' \in R_v$ such that $u \in F_{a'}$ and $a = A_{a'}(u)$, and this quantity is given by the value stored in $\Delta t(u, a)$ at the beginning for Line~\ref{line:decrease_t_for_loop}. By construction (see Line~\ref{line:find_Delta_t}), for any such $a'$ we have 
        \[
        t_{a'}(u) \leq \Delta t(u, a) \leq \Delta'(u, a),
        \]since $\Delta t(u, a)$ is equal to the largest summand in the sum defining $\Delta'(u, a)$. It follows that
        \begin{align*}
             \sum_{u \in V, a \in D}(t^{\mathrm{before}}(u, a) - t^{\mathrm{after}}(u, a)) \leq\, & \sum_{u \in V, a \in D} \Delta'(u, a) \\
            =\, & \sum_{u \in V, a \in D} \left(\sum_{a': a' \in R_v, u \in F_{a'}, a = A_{a'}(u)} t_{a'}(u)\right) \\
            =\, & \sum_{a' \in R_v} \sum_{u \in F_{a'}} t_{a'}(u) \\
            =\, & |R_v| \cdot c_{a_0} \leq |D| \cdot c_{a_0}.
        \end{align*}
        Here, the last equality follows from $\sum_{u \in F_{a'}} t_{a'}(u) = c_{a_0}$ (see Line~\ref{line:find_t}). On the other hand, for $a = A'(v)$ we must have $A_a(u) = A'(u)$ for every $u \in F_a$, and therefore we have
        \begin{align*}
        \sum_{u \in V}(t^{\mathrm{before}}(u, A'(u)) - t^{\mathrm{after}}(u, A'(u))) & \geq \sum_{u \in F_a} (t^{\mathrm{before}}(u, A_a(u)) - t^{\mathrm{after}}(u, A_a(u))) \\
        & = \sum_{u \in F_a} \Delta t(u, A_a(u)) \\ 
        & \geq \sum_{u \in F_a} t_a(u) = c_{a_0}.
        \end{align*}
        This establishes the claim.
    \end{proof}

\begin{theorem}
    Algorithm~\ref{alg:greedy} returns a satisfying assignment $A: V \to D$ whose cost is at most $|D|$ times the optimal cost.
\end{theorem}
\begin{proof}
    The algorithm returns a satisfying assignment by Lemma~\ref{lem:greedy_correctness}. 
    
    Let $A': V \to D$ be an arbitrary assignment. We have the following chain of inequalities: 
    \begin{align*}
        \sum_{v \in V} \cost(v, A(v)) & \leq \sum_{v \in V} \left(\cost(v, A(v)) - t^{\mathrm{end}}(v, A(v))\right) & \text{(Observation~\ref{obs:greedy})} \\
        & \leq \sum_{v \in V}\sum_{a \in D} \left(\cost(v, a) - t^{\mathrm{end}}(v, a)\right) & \text{(Observation~\ref{obs:monotone})} \\
        & \leq |D| \cdot \sum_{v \in V} \left(\cost(v, A'(v)) -t^{\mathrm{end}}(v, A'(v))\right) & \text{(Claim~\ref{claim:approx_ratio}, summing over all iterations)} \\
        & \leq |D| \cdot \sum_{v \in V} \cost(v, A'(v)) & \text{(Observation~\ref{obs:nonnegative})} 
    \end{align*}
    Since $A'$ is arbitrary, the theorem follows by taking $A'$ to be an optimal assignment.
\end{proof}

\subsection{An LP-based algorithm}\label{subsec:LP}

We now present another approximation algorithm that is based on the \emph{basic linear programming} (BLP) relaxation for the $\MinCostCSP$ problem \cite{sherali1990hierarchy,kumar2011lp}. Given an $\MinCostCSP(\Gamma)$ instance $I = (V, \mathcal{C}, \cost)$, its BLP relaxation can be formulated as follows.

\begin{figure}[ht]
    \centering
\begin{alignat}{4}
&\text{minimize} &\qquad& \sum_{v \in V}\sum_{a \in D} p_{v, a} \cdot \cost(v, a)\notag \\ 
&\text{subject to} &\qquad& \sum_{a \in D} p_{v, a} = 1, \qquad &&\forall v \in V,\\
& &\qquad& \sum_{x \in R} p_{C, x} = 1, \qquad &&\forall C = (R, S) \in \mathcal{C},\\
& &\qquad& \sum_{x \in R} p_{C, x} \cdot \mathbf{1}[x(v) = a]= p_{v, a}, \qquad &&\forall C = (R, S) \in \mathcal{C},\,\forall v \in S,\,\forall a \in D, \\
& &\qquad& p_{v, a}, \, p_{C, x} \geq 0, \qquad &&\forall v \in V, a \in D, C = (R, S) \in \mathcal{C}, x \in R.
\end{alignat}
    \caption{BLP relaxation for $\MinCostCSP(\Gamma)$}
    \label{fig:lp}
\end{figure}

Recall that each constraint $C$ is represented by a pair $(R, S)$ where $R$ is some $k$-ary relation and $S$ is a $k$-tuple of variables to which $R$ is applied. Here in the LP formulation we abuse the notation slightly and think of each satisfying tuple $x \in R$ also as a partial assignment $S \to D^{k}$, and $x(v)$ is the label that $x$ assigns to $v$. Informally, the linear program maintains a distribution of satisfying partial assignments for each constraint, and it requires that if a variable appears in multiple constraints, then the marginal distribution on this variable should be consistent across these constraints. The objective function to minimize is then the expected cost under these marginal distributions. It is easy to see that any integral assignment to $I$ corresponds to distributions supported on a single element, so this is indeed a relaxation of the original problem. We remark that in polynomial time we can actually enforce the marginal consistency requirement on any constant-sized set of variables, and thus obtain tighter relaxations on the so-called \emph{Sherali-Adams hierarchy} (\cite{sherali1990hierarchy}). However, this will not be needed for our purpose.

For any \MinCostCSP\ instance $I$, let $\LP(I)$ be the optimal value of the BLP relaxation for $I$. Our LP-based algorithm is as follows. Note that as before, we assume that $I$ is given as a non-trivial $(2,3)$-minimal instance.

\begin{algorithm}
    \caption{LP-based algorithm for \MinCostCSP}\label{alg:LP}
 \hspace*{\algorithmicindent} \textbf{Input:} $I = (V, \{R_u\}_{u \in V}, \{R_{u,v}\}_{u,v\in V})$ a (2,3)-minimal instance \\
 \hspace*{\algorithmicindent} \textbf{Output:}  $A: V \to D$ a satisfying assignment for $I$. 
\begin{algorithmic}[1]
\Procedure{MinCostHom-LP}{$I$}
    \State Solve the Basic LP relaxation and obtain $\{p_{v, a} \mid (v, a) \in V \times D\}$.
    \For {$v \in V$}
        \State $R_v' \gets \{a \in R_v \mid p_{v, a} \geq 1/|D|\}$
    \EndFor
    \For {every distinct $u, v \in V$}
        \State $R'_{u, v} \gets R_{u, v} \cap (R_u' \times R_v')$
    \EndFor
    \State Return any satisfying solution to the $(2,3)$-minimal instance $I' = (V, \{R_u'\}_{u \in V}, \{R_{u,v}'\}_{u,v\in V})$.
\EndProcedure
\end{algorithmic}
\end{algorithm}

In words, we remove labels that receive tiny LP probabilities (less than $1 / |D|$) and solve the new instance. Observe that $R_v'$ is nonempty for every $v \in V$, since at least one of the $|D|$ labels will get a probability that is at least $1/|D|$.

\begin{theorem}
    Let $I$ be a nontrivial (2,3)-minimal instance. Then on input $I$, Algorithm~\ref{alg:LP} returns in polynomial time a satisfying assignment to $I$ whose cost is at most $|D| \cdot \LP(I)$.
\end{theorem}
\begin{proof}
    We first show that $I'$ is indeed a (2,3)-minimal instance. We verify the conditions in Definition~\ref{def:2-3-minimal} one by one. For readers' convenience, we restate the conditions in italic.
    \begin{enumerate}
        \item[(a)] \emph{For every distinct $u, v\in V$, $R'_{v,u} = \{(b, a) \mid (a, b) \in R'_{u, v}\}$.} This is equivalent to the statement $(a, b) \in R'_{u, v} \Leftrightarrow (b, a) \in R'_{v, u}$. By symmetry, it is sufficient to show that $(a, b) \in R'_{u, v} \Rightarrow (b, a) \in R'_{v, u}$. Let $(a,b) \in R'_{u,v} \subseteq R_{u,v}$, then $a \in R'_u$ and $b \in R'_v$. Since $I$ is (2,3)-minimal, we have that $(b,a) \in R_{v, u}$, so $(b, a) \in R_{v, u} \cap (R'_v \times R'_u) = R'_{v, u}$.
        \item[(b)] \emph{For every distinct $u, v \in V$, $R_u' = \{a \mid \exists b \in R_v', (a, b) \in R_{u,v}'\}$.} By construction we have $ \{a \mid \exists b \in R_v', (a, b) \in R_{u,v}'\} \subseteq R_u'$. To show the other inclusion, let us pick $a \in R_u'$. By Lemma~\ref{lem:01all}, we have the following two cases:
        \begin{itemize}
            \item $\{a\} \times R_v \subseteq R_{u, v}$. In this case, we have $\{a\} \times R_v' \subseteq R_{u, v} \cap (R_u' \times R_v') = R_{u, v}'$. Since $R_v'$ is nonempty, we have that $a \in \{a \mid \exists b \in R_v', (a, b) \in R_{u,v}'\}$.
            \item $\{a\} \times R_v \not\subseteq R_{u, v}$. In this case, there is a unique $b_0 \in R_v$ such that $(a, b_0) \in R_{u,v}$. By the LP constraint, we have $p_{v, b_0} \geq p_{u, a} \geq 1/|D|$, and therefore $b_0 \in R_v'$ and $a \in \{a \mid \exists b \in R_v', (a, b) \in R_{u,v}'\}$.
        \end{itemize}
        \item[(c)] \emph{For every pairwise distinct $u,v,w \in V$ and $(a, b) \in R'_{u,v}$, there exists $c \in R'_w$ such that $(a, c) \in R'_{u,w}$ and $(b, c) \in R'_{v,w}$. } Similar to part (b), we again have two cases:
        \begin{itemize}
            \item $\{a\} \times R_w \not\subseteq R_{u, w}$ or $\{b\} \times R_w \not\subseteq R_{v, w}$. Since $I$ is (2,3)-minimal, there exists $c \in R_w$ such that $(a, c) \in R_{u,w}$ and $(b, c) \in R_{v,w}$. By Lemma~\ref{lem:01all}, $c$ is either the unique element in $R_w$ such that $(a, c) \in R_{u, w}$, or the unique element in $R_w$ such that $(b, c) \in R_{v, w}$. In either case, we can conclude that $p_{w, c} \geq \min(p_{u, a}, p_{v, b}) \geq 1 / |D|$, so $c \in R_w'$ as required. 
            \item $\{a\} \times R_w \subseteq R_{u, w}$ and $\{b\} \times R_w \subseteq R_{v, w}$. In this case, since $R'_w$ is non-empty, there exists some $c_0 \in R'_w \subseteq R_w$ so we have $(a, c_0) \in R_{u, w}'$ and $(b, c_0) \in R_{v, w}'$.
        \end{itemize}
    \end{enumerate}
    This establishes that $I'$ is (2,3)-minimal. Note that every binary relation in $I'$ is also of the three types described in Lemma~\ref{lem:01all}, so it is also preserved by the dual discriminator operation. By Theorem~\ref{thm:bounded-width}, we can find a satisfying assignment $A: V \to D$ for $I'$ in polynomial time. Note that for this assignment, we have
    \begin{align*}
        \sum_{v \in V}\cost(v, A(v)) & \leq \sum_{v \in V}|D| \cdot p_{v, A(v)} \cdot \cost(v, A(v)) \\ 
        & \leq |D| \cdot \sum_{v \in V} \sum_{a \in D} p_{v, a}\cdot\cost(v, a) \\
        & = |D| \cdot \LP(I). \qedhere
    \end{align*}
\end{proof}

\section{NU polymorphism as necessary condition for constant-factor approximability}\label{sec:nu_necessary}

In this section, we establish the following necessary condition for the constant-factor approximability of $\MinCostCSP(\Gamma)$ problem. We note that Dalmau et al.~\cite{dalmau2018towards} obtained a similar necessary condtion for the problem of MinCSP. We will closely follow their proof.

\begin{theorem}[Theorem~\ref{thm:intro_necessary} restated]\label{thm:necessary}
    Let $\Gamma$ be a constraint language such that $\MinCostCSP(\Gamma)$ has a constant-factor approximation, then $\Pol(\Gamma)$ contains a conservative NU polymorphism, unless P = NP.
\end{theorem}

\subsection{Gadget reductions and primitive positive interpretations}

Let us first establish some sufficient condition for reductions between $\MinCostCSP$s. 

Given two constraint languages $\Gamma_1$ and $\Gamma_2$, we write $\MinCostCSP(\Gamma_1) \lcf \MinCostCSP(\Gamma_2)$ if the constant-factor approximability for $\MinCostCSP(\Gamma_2)$ implies the constant-factor approximability for $\MinCostCSP(\Gamma_1)$. By definition, $\lcf$ is transitive.

\begin{definition}[pp-definition]\label{def:pp_def}
    Let $\Gamma$ be a constraint language over $D$ and $R$ a $k$-ary relation over the same domain. We say that $\Gamma$ \emph{pp-defines} $R$, if there exist some $m > 0$ and a conjunction $\mathfrak{C}$ over variables $x_1, \ldots, x_k, y_1, \ldots, y_m$ consisting of relations in $\Gamma$ and the equality relation ($\eq_D := \{(u, v) \in D^2 \mid u = v\}$) over $D$, such that
    \[
    R(x_1, \ldots, x_k) \equiv \exists y_1\cdots \exists y_m \mathfrak{C}.
    \]
    If $\Gamma'$ is another constraint language over the same domain $D$, then we say that $\Gamma$ pp-defines $\Gamma'$ if $\Gamma$ pp-defines every relation in $\Gamma'$.
\end{definition}

The following theorem establishes a connection (often referred to as \emph{Galois correspondence} in the literature) between polymorphisms and pp-definitions. We say that a relation $R$ with arity $k$ is \emph{irreducible}, if for every distinct $i, j \in [k]$, there exists $(x_1, \ldots, x_k) \in R$ such that $x_i \neq x_j$.

\begin{theorem}[\cite{bondarchuk1969galois, geiger1968closed}]\label{thm:galois}
    Let $\Gamma$ be a constraint language and $R$ some $k$-ary relation over the same domain. Then we have
    \begin{itemize}
        \item $\Pol(\Gamma) \subseteq \Pol(R)$ if and only $\Gamma$ pp-defines $R$.
        \item If $R$ is irreducible, then $\Pol(\Gamma) \subseteq \Pol(R)$ if and only $\Gamma$ pp-defines $R$ without using the equality relation.
    \end{itemize}
\end{theorem}

\begin{definition}[pp-interpretation]
Let $\Gamma_1$ and $\Gamma_2$ be constraint languages over domains $D$ and $E$ respectively. We say that $\Gamma_1$ pp-interprets $\Gamma_2$ if there exist $n \in \mathbb{N}$, $F \subseteq D^n$, and a surjective function $f: F \to E$ such that $\Gamma_1$ pp-defines the following relations:
\begin{itemize}
    \item $F$ as an $n$-ary relation over $D$.
    \item For every $R \in \Gamma_2$ with some arity $k$, the relation 
    \[
    f^{-1}(R) = \{(x^{(1)}, x^{(2)}, \ldots, x^{(k)}) \in D^{kn} \mid x^{(i)} \in F \text{ for } i = 1, 2, \ldots, k, (f(x^{(1)}), \ldots, f(x^{(k)})) \in R\}
    \]
    Here each $x^{(i)}$ is a $n$-tuple over $D$ and we are thinking of $(x^{(1)}, x^{(2)}, \ldots, x^{(k)})$ as a flattened $kn$-tuple over $D$.
    \item The relation 
    \[
    f^{-1}(\mathrm{eq}_{E}) = \left\{(x^{(1)}, x^{(2)}) \in D^{2n} \mid x^{(i)} \in F \text{ for } i = 1, 2,\, f(x^{(1)}) = f(x^{(2)})\right\}
    \]
    Here again $x^{(1)}$ and $x^{(2)}$ are $n$-tuples over $D$ and $(x^{(1)}, x^{(2)})$ is a flattened $2n$-tuple.
\end{itemize}
\end{definition}

We say that $\Gamma_1$ pp-interpretes $\Gamma_2$ in the first power if $n = 1$ in the above definition. It is known in the case of standard decision $\CSP$ that the existence of a pp-interpretation implies a gadget reduction, where we simply replace constraints in $\Gamma'$ with constraints in $\Gamma$ using the pp-definitions. However, in the case of $\MinCostCSP$, for the purpose of the reduction we would also need to translate the costs between the two instances. This is straightforward for $n=1$, but there seems to be no natural way of doing this if we are using a pp-interpretation with $n \geq 2$. 

\begin{lemma}\label{lem:pp-interpretation-reduction}
    Let $\Gamma_1$ be a constraint language over $D$ and  $\Gamma_2$ a constraint language over $E$. If $\Gamma_1$ pp-interpretes $\Gamma_2$ in the first power and one of the following holds:
    \begin{itemize}
        \item The equality relation $\mathrm{eq}_D$ over $D$ is in $\Gamma_1$.
        \item Every $R \in \Gamma_2$ is irreducible.
    \end{itemize}
    Then $\MinCostCSP(\Gamma_2) \leq_\CF \MinCostCSP(\Gamma_1)$ 
\end{lemma}
\begin{proof}
    Let $F \subseteq D$ and $f: F \to E$ be as in the definition of pp-interpretation. Let $I_2 = (V, \mathcal{C}, \cost)$ be a $\MinCostCSP(\Gamma_2)$ instance. We define a $\MinCostCSP(\Gamma_1)$ instance $I_1 = (V, \mathcal{C}', \cost')$ as follows: 
    \begin{itemize}
        \item $I_1$ has the same set of variables $V$ as $I_2$.
        \item For each constraint $C = (R, S) \in \mathcal{C}$, we would like to add a constraint $C' = (f^{-1}(R), S)$ to $\mathcal{C}'$. To do this, without loss of generality assuming $S = \{x_1, \ldots, x_k\}$, we take the pp-definition of $f^{-1}(R)$, which is of the form
        \[
        \exists y_1 \cdots \exists y_m \mathfrak{C}.
        \]
        Here $y_1, \ldots, y_m$ are auxiliary variables with zero costs and $\mathfrak{C}$ is a conjunction of constraints, each being either a relation from $\Gamma_1$ or $\eq_D$ applied to some of the variables in $\{x_1, \ldots, x_k, y_1, \ldots, y_m\}$. Now if $\eq_D \in \Gamma_1$, then we can think of $\mathfrak{C}$ as a $\CSP(\Gamma_1)$ instance with variables being $\{x_1, \ldots, x_k, y_1, \ldots, y_m\}$, and this instance can be satisfied if and only $(x_1, \ldots, x_k) \in f^{-1}(R)$, so we add (the constraints and the auxiliary variables of ) this instance to $I_2$. If $\eq_D \not\in \Gamma_1$, but every $R \in \Gamma_2$ is irreducible, then by Theorem~\ref{thm:galois}, we may assume that $\mathfrak{C}$ does not contain $\eq_D$ and therefore we can still write it as an instance of $\CSP(\Gamma_1)$ and add it to $I_2$.
        \item For each variable $x \in V$ and $a \in D$, define $\cost'(x, b) = \left\{\begin{array}{ll}
            \cost(x, a) & \text{if } b \in f^{-1}(a) \text{ for some } a \in E, \\
            +\infty & \text{otherwise.}
        \end{array}\right.$
    \end{itemize}

    Clearly, for every satisfying assignment $A: V \to E$ to $I_2$, we may define $A': V \to F$ such that $A'(x) \in f^{-1}(A(x))$, and $A'$ will be a satisfying assignment to $I_1$ with the same cost. In particular, this means that $\Opt(I_1) \leq \Opt(I_2)$.

    Now if we have a constant-factor approximation algorithm for $\MinCostCSP(\Gamma_1)$, we can use it to obtain a solution $A_1: V \to D$ such that $\cost_{I_1}(A_1) \leq t \cdot \Opt(I_1)$ for some constant $t$ independent of $I_1$. In fact, we may assume $A_1: V \to F$, since every label not in $F$ has infinite cost. Take $A_2: V \to E, x \mapsto f(A_1(x))$, then by construction $A_2$ is a satisfying assignment for $I_2$, and
    \[
    \cost_{I_2}(A_2) =\cost_{I_1}(A_1) \leq t \cdot \Opt(I_1) \leq t \cdot \Opt(I_2).
    \]
    Thus we obtain a constant-factor approximation algorithm for $\MinCostCSP(\Gamma_2)$ as well.
\end{proof}

We refer interested readers to the survey by Barto el al.~\cite{barto_et_al:DFU.Vol7.15301.1} which contains a more detailed exposition on pp-interpretation (and its generalization \emph{pp-construction}) in the context of decision CSPs.

\subsection{Proof of Theorem~\ref{thm:necessary}}

The proof contains two cases: either $\Gamma$ has unbounded width, or it has bounded width. We use the following characterization for bounded-widthness of constraint languages. 

\begin{theorem}[\cite{barto2014constraint, dalmau2013robust}]\label{thm:unbounded-abelian}
    Let $\Gamma$ be a constraint language that contains all singleton relations. $\Gamma$ is not bounded width if and only if there exists some nontrivial finite abelian group $G$ such that $\Gamma$ pp-interprets $\Gamma(G)$ in the first power using pp-definitions without equality.
\end{theorem}

Here $G$ being nontrivial means it has at least 2 elements, and $\Gamma(G)$ is the set of relations $\{R_{abc} = \{(x, y, z) \in G^3 \mid ax + by + cz = 0\} \mid a,b,c\in \mathbb{Z}\}$\footnote{Here $ax$ denotes the sum of $a$ copies of $x$. Note that this is a finite set of relations, since $G$ is a finite group.} over $G$. 

In the unbounded-width case we shall use a reduction from the Nearest Codeword problem, and in the bounded-width case we reduce from the hypergraph vertex cover problem. 

\begin{definition}
    In the Nearest Codeword problem over a finite field $\mathbb{F}_p$, we are given a matrix $A \in \mathbb{F}_p^{m \times n}$ and a vector $x \in \mathbb{F}_p^n$, and we are asked to find a vector $y \in \mathbb{F}_p^n$ such that $Ay = 0$ and the number of nonzero entries in $x - y$ is minimized.
\end{definition}

\begin{definition}
    In the $k$-uniform hypergraph vertex cover problem, we are given a $k$-uniform hypergraph (namely, each hyperedge is contains $k$ vertices), and our goal is to choose a minimum number of vertices so that from each hyperedge we have chosen at least one vertex. 
\end{definition}

The following theorems give the best known NP-hardness results for approximating these two problems. 

\begin{theorem}[\cite{dumer2003hardness,cheng2012deterministic}]\label{thm:ncp}
    The Nearest Codeword problem over any finite field $\mathbb{F}_p$ is NP-hard to approximate within a factor of $2^{\log^{1-\epsilon}(n)}$, for any constant $\epsilon > 0$.
\end{theorem}

\begin{theorem}[\cite{dinur2005new}]\label{thm:hvc}
    The $k$-uniform hypergraph vertex cover problem is NP-hard to approximate within a factor of $k-1- \epsilon$, for any $k \geq 3$ and $\epsilon > 0$.
\end{theorem}

We remark that if we further assume the Unique Games Conjecture, then we can improve the hardness factor for $k$-uniform hypergraph vertex cover from $k-1-\epsilon$ to $k-\epsilon$ \cite{bansal2010inapproximability}, but this difference does not matter for us here.

The following is a simple corollary from the hardness of the Nearest Codeword problem.

\begin{corollary}
    Let $\Gamma_p$ be the set of all relations of the form $R_{abc} = \{(x, y, z) \in \mathbb{F}_p^3 \mid ax + by + cz = 0\}$\footnote{Here $ax$ denotes the $\mathbb{F}_p$ multiplication between field elements $a, x \in \mathbb{F}_p$.} where $a, b, c \in \mathbb{F}_p \backslash \{0\}$ over some finite field $\mathbb{F}_p$. Then $\MinCostCSP(\Gamma_p)$ is NP-hard to approximate within a factor of $2^{\log^{1-\epsilon}(n)}$.
\end{corollary}
\begin{proof}
    We show how to cast the Nearest Codeword problem over $\mathbb{F}_p$ as a $\MinCostCSP(\Gamma_p)$ instance. We first add the variables $y_1, \ldots, y_n$ denoting entries of $y$. For each $y_i$, we impose a cost of 1 if it is not equal to $x_i$, and 0 if it is equal to $x_i$. This models the minimum Hamming weight requirement. To model the constraint $Ay = 0$, we first note that $\Gamma_p$ can also be used to simulate $R_{abc} = \{(x, y, z) \in \mathbb{F}_p^3 \mid ax + by + cz = 0\}$ if one of $a, b, c$ is zero. This can be achieved as follows: to obtain the constraint $ax + by = 0$, we create a dummy variable $z$ and add the constraint $R_{abc}(x, y, z)$ for an arbitrary nonzero $c$, and then we set the non-zero label costs for $z$ to be all infinite and its zero label cost to be just 0, effectively forcing this variable to be zero. This can be extended easily if two of $a, b, c$ are zeros.
    
    Now we need to transform each linear equation of the form $\sum_{i = 1}^m A_{ki}y_i = 0$ so that left hand side contains at most 3 variables. This can be done via the standard trick: by introducing a new auxiliary variable $z$ (which has zero cost for any label), we may rewrite $\sum_{i = 1}^m A_{ki}y_i = 0$ equivalently as $A_{k1}y_1 + A_{k2}y_2 + z = 0$ and $-z + \sum_{i = 3}^m A_{ki}y_i = 0$, and thereby reducing the number of variables on the left hand side by 1. We repeat this procedure until all equations have at most 3 variables on the left hand side. It is clear that the resulting instance is a $\MinCostCSP(\Gamma_p)$ instance which is equivalent to the original Nearest Codeword instance. By Theorem~\ref{thm:ncp}, we can therefore conclude that $\MinCostCSP(\Gamma_p)$ is also NP-hard to approximate within a factor of $2^{\log^{1-\epsilon}(n)}$.
\end{proof}

\begin{lemma}\label{lem:cf_bdwidth}
    Let $\Gamma$ be a constraint language such that $\MinCostCSP(\Gamma)$ has a constant-factor approximation, then $\Gamma$ is bounded-width unless P = NP.
\end{lemma}
\begin{proof}
    We prove the contrapositive. Let $\Gamma$ be a constraint language with unbounded width. Without loss of generality, assume that $\Gamma$ contains all unary relations. Then by Theorem~\ref{thm:unbounded-abelian}, there exists some nontrivial finite abelian group $G$ such that $\Gamma$ pp-interprets $\Gamma(G)$ in the first power. In particular, $\Gamma$ pp-interprets the set of all irreducible relations $\Gamma(G)^{\mathrm{irr}} \subseteq \Gamma(G)$ in the first power. By Lemma~\ref{lem:pp-interpretation-reduction}, we have $\MinCostCSP(\Gamma(G)^{\mathrm{irr}}) \lcf \MinCostCSP(\Gamma)$.
    
    We now claim that $\MinCostCSP(\Gamma(G)^{\mathrm{irr}})$ contains $\MinCostCSP(\Gamma_p)$ as a special case for some prime $p$. Note that $G$ must contain some cyclic subgroup of prime order: one may find this subgroup by taking a subgroup of prime order of some cyclic subgroup generated by a single nonzero element in $G$. We may identify this subgroup of order $p$ with the finite field $\mathbb{F}_p$. Note that any relation $R_{abc} = \{(x, y, z) \in \mathbb{F}_p^3 \mid ax + by + cz = 0\}$ over $\mathbb{F}_p$ is irreducible if $a, b, c$ are all nonzero, so these relations are contained in $\Gamma(G)^{\mathrm{irr}}$. So we have that $\MinCostCSP(\Gamma(G)^{\mathrm{irr}})$ contains $\MinCostCSP(\Gamma_p)$ as a subproblem (where we set the cost of any label outside $\mathbb{F}_p$ to be infinite). It follows that $\MinCostCSP(\Gamma(G)^{\mathrm{irr}})$, and therefore $\MinCostCSP(\Gamma)$, do not have a constant-factor approximation unless P = NP.
\end{proof}

For the bounded-width case, we use the following reduction from the hypergraph vertex cover problem.

\begin{lemma}[\cite{dalmau2018towards}]\label{lem:nu_bdw_hvc}
    Let $\Gamma$ be a bounded-width constraint language which is not preserved by any NU operation. If $\Gamma$ contains all unary singleton relations, then for every $k \geq 1$, there is a $k$-ary relation $R$ pp-definable from $\Gamma$ and $a, b \in D$ such that 
    \[
    R \cap \{a, b\}^k = \{a, b\}^k \backslash \{(a, a, \ldots, a)\}.
    \]
\end{lemma}

\begin{lemma}\label{lem:bdwidth-nu}
    Let $\Gamma$ be a bounded-width constraint language which is not preserved by any NU operation. Then $\MinCostCSP(\Gamma)$ does not have a constant-factor approximation unless P = NP.
\end{lemma}
\begin{proof}
    Assume that $\Gamma$ has all unary relations without loss of generality. Note that the $k$-uniform hypergraph vertex cover problem is the $\MinCostCSP$ problem with a single relation $R_k = \{0, 1\}^k \backslash \{(0,0,\ldots,0)\}$, which is pp-definable from $\Gamma$ by Lemma~\ref{lem:nu_bdw_hvc} (by thinking of $a$ as 0 and $b$ as 1, and the assumption that $\{a, b\}$ as a unary relation is in $\Gamma$). Observe that $R_k$ is irreducible, so the reduction in Lemma~\ref{lem:pp-interpretation-reduction} implies that $\MinCostCSP(\Gamma)$ is as hard to approximate as the $k$-uniform hypergraph vertex cover problem, in particular, by Theorem~\ref{thm:hvc}, it is NP-hard to approximate $\MinCostCSP(\Gamma)$ within a factor of $k - 1 - \epsilon$ for any $\epsilon > 0$. Since $k$ can be arbitrarily large, this implies that $\MinCostCSP(\Gamma)$ does not have a constant-factor approximation, unless P = NP.
\end{proof}

Theorem~\ref{thm:necessary} can now be obtained by combining these two cases.
\begin{proof}[Proof of Theorem~\ref{thm:necessary}]
    Let $\Gamma$ be a constraint language such that $\MinCostCSP(\Gamma)$ has a constant-factor approximation and assume that P $\neq$ NP, then by Lemma~\ref{lem:cf_bdwidth}, $\Gamma$ must be bounded-width. It then follows from Lemma~\ref{lem:bdwidth-nu} that $\Gamma$ must be preserved by some NU operation.
\end{proof}

\begin{remark}\label{remark:boolean_nu}
In the Boolean case ($|D| = 2$), it follows from Khanna et al.'s classification~\cite{khanna2001approximability} as well as Post's classification of Boolean clones~\cite{post1941two} that the necessary condition of having a conservative NU polymorphism is sufficient as well. Recall that we have the following three classes of constant-factor approximable Boolean $\MinCostCSP$s:
\begin{itemize}
    \item $\Gamma$ can be expressed as a 2CNF-formula. In this case, $\Gamma$ is preserved by the (unique) majority operation. (This case also includes constraint languages that are width-2 affine for which $\MinCostCSP$ can be solved to optimality.)
    \item $\Gamma$ is expressible as a CNF formula where each clause is of the form $x_1 \vee \cdots \vee x_k$, $\neg x_1 \vee x_2$, or $\neg x_1$ where $k \leq K$ for some $K$ depending on $\Gamma$. In this case $\Gamma$ is preserved by the $(K+1)$-ary NU operation $\mathrm{th}_{2}^{K+1}$, where 
    \[
    \mathrm{th}_{p}^{n}(x_1, \ldots, x_n) = \left\{\begin{array}{ll}
        1 & \text{if } |\{i \in [n] \mid x_i = 1\}| \geq p, \\
        0 & \text{otherwise.}
    \end{array}\right.
    \]
    \item $\Gamma$ is expressible as a CNF formula where each clause is of the form $\neg x_1 \vee \cdots \vee \neg x_k$, $ x_1 \vee \neg x_2$, or $x_1$ where $k \leq K$ for some $K$ depending on $\Gamma$. In this case $\Gamma$ is preserved by the $(K+1)$-ary NU operation $\mathrm{th}_{K}^{K+1}$.
\end{itemize}
It can be easily verified using Post's Lattice~\cite{post1941two} that any constraint language whose polymorphism clone contains an NU operation can be reduced to one of the three cases above. However, as soon as $|D| \geq 3$, the condition of being preserved by some NU operation is no longer sufficient (for example, see Theorem~\ref{thm:counterexample} in the following subsection).
\end{remark}

\subsection{A hard predicate with a majority polymorphism}\label{sec:counter_example}

We now present a binary relation $P_H$ which has a conservative majority polymorphism, but nonetheless $\MinCostCSP(P_H)$ is hard to approximate within any constant factor, unless UGC fails. This implies that the existence of an NU polymorphism is in general not sufficient for constant-factor approximability assuming UGC.
 
\begin{definition}
    Let $P_H$ be the binary relation on domain $A = \{0, 1, 2\}$ such that $P_H(x, y)$ holds if and only if $x \neq y$ or $x = y = 2$. 
\end{definition}

The constraint satisfaction problem defined by $P_H$ is equivalent to the graph homomorphism problem to the undirected graph shown in Figure~\ref{fig:H}. Intuitively, $P_H$ is the XOR predicate with a ``wildcard'' element 2 such that the predicate is also satisfied if some input is 2.

\begin{figure}[ht]
    \centering

    \begin{tikzpicture}[
      mycircle/.style={
         circle,
         draw=black,
         fill=gray,
         fill opacity = 0,
         text opacity=1,
         inner sep=0pt,
         minimum size=20pt,
         font=\small},
      node distance=0.6cm and 1.2cm
      ]
      \tikzset{every loop/.style={}}
      \node[mycircle]  at (0, 0) (u0) {$0$};
      \node[mycircle]  at (2, 0) (u1) {$1$};
      \node[mycircle]  at (1, {sqrt(3)}) (u2) {$2$};

    \foreach \i/\j/\p in {
      u0/u1/above,
      u0/u2/above,
      u1/u2/above}
       \draw [] (\i) -- node[font=\small] {} (\j);

\draw (u2) edge[loop above] node{} (u2);
    \end{tikzpicture}
    
    \caption{The undirected graph $H$ corresponding to $P_H$.}
    \label{fig:H}
\end{figure}
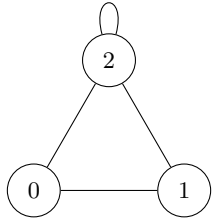

We now verify that $P_H$ is preserved by a conservative majority operation.
\begin{claim}
    Let $f: A^3 \to A$ be defined as follows
    \[
    f(a_1, a_2, a_3) = \left\{\begin{array}{ll}
       a  &  \text{if } \left|\{i \in [3] \mid a_i = a\}\right| \geq 2, \\
       2  &  \text{otherwise.}
    \end{array}\right.
    \]
    Then $f \in \Pol(P_H)$.
\end{claim}
\begin{proof}
   Let $(a_1, b_1), (a_2, b_2), (a_3, b_3) \in P_H$. We need to verify that $(f(a_1, a_2, a_3), f(b_1, b_2, b_3)) \in P_H$. This is always true if at least one of $f(a_1, a_2, a_3)$ and $f(b_1, b_2, b_3)$ is 2. If neither is 2, then there is a majority in $(a_1, a_2, a_3)$ as well as in $(b_1, b_2, b_3)$. It follows by the pigeonhold principle that there must be some $i \in [3]$ such that $a_i$ is equal to the majority element in $(a_1, a_2, a_3)$ and $b_i$ is equal to the majority element in $(b_1, b_2, b_3)$, so we have $(f(a_1, a_2, a_3), f(b_1, b_2, b_3)) = (a_i, b_i) \in P_H$.
\end{proof}
Note that $f$ is conservative since when there isn't a majority we must have $\{a_1, a_2, a_3\} = \{0, 1, 2\} \ni 2$.

To prove that $\MinCostCSP(P_H)$ is hard to approximate, we use a reduction from the \MinUnCut{} problem.

\begin{definition}
    In the \MinUnCut{} problem, the input is a weighted undirected graph $G = (V, E, \{w_e\}_{e \in E})$ where $w_e \geq 0$ for every $e \in E$, and we are asked to remove a subset $E' \subseteq E$ of the edges such that the remaining graph $G' = (V, E \,\backslash\, E')$ is bipartite. The goal is to minimize the total weight of removed edges $\sum_{e \in E'}w_e$.
\end{definition}

We use $\Opt(G)$ to denote the value of an optimum solution to \MinUnCut{}($G$). Without loss of generality, we may assume that the total edge weight in a \MinUnCut{} instance is normalized to be 1, i.e., $\sum_{e \in E} w_e = 1$. 

\begin{theorem}[\cite{khot2007optimal}]\label{thm:minuncut_hardness}
    Assuming UGC, there exists some constant $c > 0$ such that for all sufficiently small $\epsilon > 0$ it is NP-hard to distinguish instances of \MinUnCut{} with value at most $\epsilon$ and instances with value at least $c\cdot \sqrt{\epsilon}$. In particular, it is NP-hard to approximate \MinUnCut{} within any constant factor, assuming UGC. 
\end{theorem}

\begin{theorem}\label{thm:counterexample}
    Assuming UGC, it is NP-hard to approximate $\MinCostCSP(P_H)$ within any constant factor.
\end{theorem}
\begin{proof}
    Given any \MinUnCut{} instance $G = (V, E, \{w_e\}_{e \in E})$, we construct a \MinCostCSP($P_H$) instance $I$ such that $\Opt(G) = \Opt(I)$. This reduction combined with Theorem~\ref{thm:minuncut_hardness} will establish our theorem. The reduction is as follows. The variable set of $I$ will be $V \cup \{z_e, z_e' \mid e \in E\}$ where we take vertices in $G$ plus two distince auxiliary variables $z_e, z_e'$ for every edge $e \in E$. For every $e = \{x, y\}\in E$, we add three constraints $P_H(x, z_e), P_H(z_e, z_e'), P_H(z_e', y)$ to $I$ (note that the order of $x$ and $y$ does not matter). For the cost function $c$, we define $c(x, 0) = c(x, 1) = 0$, $c(x, 2) = 1$ for every $x \in V$, and $c(z_e, 0) = c(z_e, 1) = c(z_e', 0) = c(z_e', 1) = 0$, $c(z_e, 2) = c(z_e', 2) = w_e$ for the auxiliary variables. This completes the construction. See Figure~\ref{fig:reduction} for an illustration.

    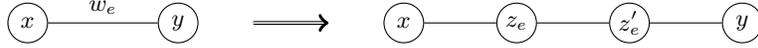
\begin{figure}[t]
    \centering

    \begin{tikzpicture}[
      mycircle/.style={
         circle,
         draw=black,
         fill=gray,
         fill opacity = 0,
         text opacity=1,
         inner sep=0pt,
         minimum size=15pt,
         font=\small},
      node distance=0.6cm and 1.2cm
      ]
      \tikzset{every loop/.style={}}
      \node[mycircle]  at (0, 0) (u) {$x$};
      \node[mycircle]  at (2, 0) (v) {$y$};
      
\draw[double, ->] (3, 0) -- (4, 0);

      \node[mycircle]  at (5, 0)  (u1)  {$x$};
      \node[mycircle]  at (6.5, 0) (z1) {$z_e$};
      \node[mycircle]  at (8, 0) (z2) {$z_e'$};
      \node[mycircle]  at (9.5, 0) (v1) {$y$};

    \foreach \i/\j/\t/\p in {
      u/v/$w_e$/above,
      u1/z1//above,
      z1/z2//above,
      z2/v1//above}
       \draw [] (\i) -- node[font=\small,above] {\t} (\j);
    \end{tikzpicture}
    
    \caption{The reduction from \MinUnCut{} to \MinCostCSP$(P_H)$. The nonzero costs are $c(x, 2) = c(y, 2) = 1$, $c(z_e, 2) = c(z_e', 2) = w_e$.}
    \label{fig:reduction}
\end{figure}

    We claim that $\Opt(G) = \Opt(I)$. We first show that $\Opt(G) \geq \Opt(I)$. Take any optimal assignment for $G$. We take the same assignment for the vertex variables in $I$ which generates no cost. For any edge $e = \{x, y\}$ that is satisfied by the assignment, we can set $z_e = 1 - x$, $z_e' = 1 - y$ and satisfy all three constraints $P_H(x, z_e), P_H(z_e, z_e'), P_H(z_e', y)$ with no cost. For any edge $e = \{x, y\}$ that is not satisfied, we can set $z_e = 2$ and $z_e' = 1-y$, satisfying all three constraints $P_H(x, z_e), P_H(z_e, z_e'), P_H(z_e', y)$ with cost $w_e$. So we obtain an assignment for $I$ that has value $\Opt(G)$.

    The other direction can be shown similarly. First observe that for each edge $e$ we may assume at most one of the two auxiliary variables $z_e, z_e'$ is set to 2. Also, since for any vertex variable $x$ we have $c(x, 2) = 1 = \sum_{e \in E} w_e$, we may assume that no vertex variable $x$ is set to 2. Now if we take an optimal assignment $A$ for $I$ with these two assumptions, $A$ restricted to the vertex variables is a valid assignment for $G$ whose the total weight of violated edges is at most the cost of $A$, which implies that $\Opt(G) \leq \Opt(I)$. 
\end{proof}

\section{Application: Dichotomy for MinCostCSP with permutation constraints}\label{sec:application}

As an application of our results, we give a complete classification for $\MinCostCSP(\Gamma)$ where $\Gamma$ contains all \emph{permutation relations}.
\begin{definition}
    A binary relation $R \subseteq D^2$ over $D$ is called a \emph{permutation relation} if $R = \{(a, \sigma(a)) \mid a \in D\}$ for some bijective $\sigma: D \to D$.
\end{definition}

\begin{theorem}[Theorem~\ref{thm:intro_perm} restated]\label{theorem:dichotomy_perm}
    Let $\Gamma$ be a set of relations over $D$ such that it contains all permutation relations. Then $\MinCostCSP(\Gamma)$ is $|D|$-approximable if $\Gamma$ is preserved by a conservative majority operation. Otherwise, if $\Gamma$ is not preserved by any conservative majority operation, then it is also not preserved by any conservative NU operation and $\MinCostCSP(\Gamma)$ is not constant-factor approximable, assuming P $\neq$ NP.
\end{theorem}

A constraint language that contains all permutation relations can be seen as a natural generalization of Boolean constraint languages that allow negation of variables. Our classification relies on the classification of \emph{homogeneous algebras}. To state the result, we first need some definitions.

\begin{definition}
    An algebra $(D, \mathcal{F})$ consists of a set $D$ (called the \emph{universe}) and a set of operations $\mathcal{F}$ (called the \emph{basic operations}) which are functions from finite powers of $D$ to $D$. The symbols and arities of the basic operations are called the \emph{signature} of $(D, \mathcal{F})$. A \emph{term operation} is an operation obtained by composition of operations in $\mathcal{F}$. 
\end{definition}

The set of all term operations of a given algebra $(D, \mathcal{F})$ form a clone (recall Definition~\ref{def:clone}). We denote this clone by $\langle \mathcal{F} \rangle$. When $\mathcal{F} = \{s_1, \ldots, s_k\}$ consists of finitely many operations, we may also write $\langle s_1, \ldots, s_k \rangle$ in place of $\langle \mathcal{F} \rangle$.

\begin{definition}
    Let $(D, \mathcal{F})$ and $(D', \mathcal{F}')$ be two algebras with the same signature. A function $f: D \to D'$ is called a homomorphism from  $(D, \mathcal{F})$ to $(D', \mathcal{F}')$, if $f$ commutes with all basic operations. That is, for every $k$-ary function symbol $t$ in the signature, we have $t_{D'}(f(a_1), \ldots, f(a_k)) = f(t_D(a_1, \ldots, a_k))$, where $t_D$ and $t_{D'}$ are the functions $t$ represents in $(D, \mathcal{F})$ and $(D', \mathcal{F}')$ respectively. When $(D, \mathcal{F})=(D', \mathcal{F}')$, we also say that $f$ is an \emph{automorphism}.
\end{definition}

\begin{definition}
    An algebra $(D, \mathcal{F})$ is called a homogeneous algebra if every bijection $D \to D$ is an automorphism.
\end{definition}

The following claim follows directly from the definition.

\begin{claim}
    Let $\Gamma$ be a constraint language which contains all permutation relations. Then $(D, \Pol(\Gamma))$ is a homogeneous algebra.
\end{claim}

The study of homogenous algebras was initiated by Marczewski~\cite{marczewski1964homogeneous}, and a complete classification was first obtained by Marchenkov~\cite{marchenkov1982homogeneous}. Dalmau used Marchenkov's result to give a complete classification for decision $\CSP$ where the constraint language contains all permutation relations~\cite{dalmau2005new}. The following theorem is taken from~\cite{szendrei1986clones} (see also~\cite{dalmau2005new}).

\begin{theorem}[Theorem 5.9 in~\cite{szendrei1986clones}]\label{thm:hom_algebra_classification}
    Let $D$ be a finite domain such that $|D| \geq 5$. Let $(D, \mathcal{F})$ be a homogeneous algebra, then either the dual discriminator operation $d$ is a term operation, or its clone of term operations $\langle \mathcal{F} \rangle$ is equal to one of the followings:
    \begin{itemize}
        \item $E_1^0 = \langle s \rangle$, $E_1^1 = \langle s, r_n \rangle$,
        \item $E_i^0 = \langle l_i \rangle$ for $2 \leq i \leq n - 1$, $E_n^0 = \mathcal{J}$.
        \item $E_i^1 = \langle l_i, r_n \rangle$ for $2 \leq i \leq n - 3$, $E_{n-2}^1 = \langle r_n \rangle$.
    \end{itemize}
    Here $\mathcal{J}$ is the clone of projection operations. $s$ is the \emph{switching} operation, defined by 
    \[
    s(x_1, x_2, x_3) = \left\{\begin{array}{ll}
        x_3 & \text{if } x_1 = x_2, \\
        x_2 & \text{if } x_1 = x_3, \\
        x_1 & \text{otherwise. }
    \end{array}\right.
    \]
    For $2 \leq k \leq n - 1$, $l_k$ is the $k$-ary \emph{near projection} operation defined by
    \[
    l_k(x_1, x_2, \ldots, x_k) = \left\{\begin{array}{ll}
        x_1 & \text{if } |\{x_1, \ldots, x_k\}| < k, \\
        x_k & \text{otherwise. }
    \end{array}\right.
    \]
    And finally, $r_n$ is the $(n-1)$-ary operation defined by
    \[
    r_n(x_1, x_2, \ldots, x_{n-1}) = \left\{\begin{array}{ll}
        x_1 & \text{if } |\{x_1, \ldots, x_{n-1}\}| < n - 1, \\
        x_n & \text{otherwise, where } x_n \in D \backslash \{x_1, \ldots, x_{n-1}\}.
    \end{array}\right.
    \]
    
\end{theorem}

    The $|D| \geq 5$ assumption is not essential. When $|D| \geq 5$, the above algebras are all distinct. When $2 \leq |D| \leq 4$, some of these algebras become non-distinct, but the only exceptional case not covered by the classification above is the Klein 4-group (the unique 4-element group with exponent 2) with the operation $x + y + z$. However, this is not a conservative operation so we may safely ignore it for our purpose.
    
We are now ready to prove Theorem~\ref{theorem:dichotomy_perm}.
\begin{proof}[Proof of Theorem~\ref{theorem:dichotomy_perm}]
    First observe that if $\Pol(\Gamma)$ contains some majority function $f$, then there must be some $i \in [3]$ such that $f(x_1, x_2, x_3) = x_i$ when $x_1, x_2, x_3$ are pairwise distinct: if not, then there exist distinct $i, j\in[3]$ and two triples $(x_1, x_2, x_3), (y_1, y_2, y_3)$ with pairwise distinct elements within each triple such that $f(x_1, x_2, x_3) = x_i, f(y_1, y_2, y_3) = y_j$. Then, let $\pi: D \to D$ be a permutation such that $y_i = \pi(x_i)$ for every $i \in [3]$, $f$ does not preserve the permutation relation $\{(a, \pi(a)) \mid a \in D\}$, which is a contradiction. By potentially permuting the input coordinates in $f$, we get that the dual discriminator operation $d$ is also contained in $\Pol(\Gamma)$, and therefore $\MinCostCSP(\Gamma)$ is $|D|$-approximable by Theorem~\ref{thm:intro_algo}.

    Now suppose $\Pol(\Gamma)$ does not contain a conservative majority function, then in particular it does not contain the dual discriminator operation. Note that since $\Gamma$ can be assumed to contain all unary operations (see Observation~\ref{obs:mincost_unary}), every polymorphism of $\Gamma$ is conservative. It is easy to see that $r_n$ is not conservative. If $\Pol(\Gamma) = E_n^0 = \mathcal{J}$, then $\CSP(\Gamma)$ is NP-complete. Furthermore, Dalmau~\cite{dalmau2005new} showed that if $\Pol(\Gamma) = E_i^0$ for some $2 \leq i \leq n-1$, then $\CSP(\Gamma)$ is also NP-complete. So by Theorem~\ref{thm:hom_algebra_classification}, the only remaining possibility is $\Pol(\Gamma) = \langle s \rangle$. However, as is observed by Dalmau~\cite{dalmau2005new}, $\langle s \rangle$ does not contain any NU operation, so by Theorem~\ref{thm:necessary}, there is no constant factor approximation for $\MinCostCSP(\Gamma)$, assuming P $\neq$ NP.
\end{proof}

\bibliography{references}
\bibliographystyle{alpha}

\end{document}